\documentclass[12pt]{amsart}
\usepackage{amsthm,amsmath,amssymb,amsfonts,latexsym,verbatim,stmaryrd,url,xcolor}

\newcommand{\abs}[1]{|#1|} 
\newcommand{\F}{\mathbb{F}}
\newcommand{\Z}{\mathbb{Z}}
\DeclareMathOperator{\gal}{Gal}
\DeclareMathOperator{\Cay}{{\mathrm{Cay}}}

\newtheorem{theorem}{Theorem}[section]

\newtheorem{lemma}[theorem]{Lemma}

\theoremstyle{remark}

\newtheorem*{concremark*}{Concluding Remarks}

\theoremstyle{definition}

\begin{document}
\title[]{
  Uniform mixing in continuous-time quantum walks 
  on oriented, nonabelian Cayley graphs
}
\author[P. Sin]{Peter Sin}
\address{Department of Mathematics\\University of Florida\\ P. O. Box 118105\\ Gainesville FL 32611\\ USA}

\thanks{This work was partially supported by a grant from the Simons Foundation (\#633214
 to Peter Sin).}
\date{}
\begin{abstract} A family  of oriented, normal, nonabelian
  Cayley graphs is presented, whose continuous-time quantum
  walks exhibit uniform mixing.
\end{abstract}
\maketitle

\section{Introduction}
Given a simple, undirected graph $X$ we can consider a quantum walk
on $X$ whose evolution is given by the matrices $U(t)=e^{-itA}$, where
$A$ is the adjacency matrix.  If instead of an undirected graph we have an oriented graph, we can
consider the quantum walk whose transition matrices are $U(t)=e^{tS}$,
where $S$ is the skew adjacency matrix,  whose $(u,v)$ entry
is $1$ if $(u,v)\in X\times X$ is an arc, $-1$ if $(v,u)$ is an arc, and $0$ otherwise.
In both cases, the walk is said to have {\it uniform mixing} (abbreviated to UM) at time $\tau$ if all the entries of $U(\tau)$ have the same absolute value.

Here is a brief survey of  examples of UM in undirected and oriented graphs.
In the undirected case, the complete graphs $K_q$, for $q=2,3,4$ exhibit UM at appropriate times, but not for $q>4$. It is also easily checked  that
Cartesian products of  graphs that have UM at the same time will also  have UM at that time. 
Most known examples are regular graphs, but  Godsil and Zhan \cite{Godsil-Zhan} found that
the irregular graph $K_{1,3}$ has UM. The paper of Moore and Russell \cite{ Moore-Russell}, where the concept of UM first appeared, showed that hypercubes have UM. Various other families of  cubelike graphs with UM are given by Chan in \cite{ChanUMCubes}, and a construction using bent
functions is given by Cao, Feng and Wan in  \cite{CaoWanFeng}. In addition
Godsil, Mullin and Roy \cite{Godsil-Mullin-Roy} determined the parameter sets of
the strongly regular graphs that have UM.
In the oriented case there  fewer examples are known. The simplest is the $4$-cycle. As in the undirected case the class of oriented graphs with UM at a given time is closed under taking Cartesian products.
Some Cayley graphs on $(\Z/4\Z)^n$ that have UM but  are not Cartesian products
were found by Godsil and Zhang in \cite{Godsil-Zhang}. In the same paper it is proved that
if an oriented Cayley graph on a finite group $G$ has UM, then $\abs{G}$ must be an even perfect square. This result comes from a connection with
Hadamard matrices that will come up in our discussion below.

Let $G$ be a finite group  and $C$ a subset. We denote by
$\Cay(G,C)$ the Cayley digraph on $G$ with connection set $C$, where
we have an arc from $g$ to $cg$ for every $c\in C$. We assume that $1\notin C$
so that $\Cay(G,C)$ has no loops. Let $C^{(-1)}$ denote the set of
inverses of elements in $C$. Then $\Cay(G,C)$ is undirected iff $C^{(-1)}=C$, and oriented iff $C^{(-1)}\cap C=\emptyset$. In both cases, $\Cay(G,C)$ is
{\it normal} if $C$ is a union of conjugacy classes.

The purpose of this note is to describe an infinite family of
oriented, normal Cayley graphs on nonabelian $2$-groups whose continuous time
quantum walks exhibit UM. We begin with a general
theory of UM on oriented Cayley graphs in \S\ref{orientedUM}. In the \S\ref{groups_graphs} we describe
the groups and graphs and then in \S\ref{proofs} we give the proofs, based
on the character theory of these groups described in \S\ref{chars}.

\section{Uniform mixing in oriented Cayley graphs}\label{orientedUM}
Let $G$ be a finite group and let $K_j$, $j=1$,\dots, $k$ be the conjugacy classes. Let $C$ be a union of conjugacy classes such that  $C^{(-1)}\cap C=\empty$. 
We shall derive conditions for uniform mixing on the  oriented, normal Cayley graph $\Cay(G,C)$, via the spectral decomposition, just like in
in the undirected case (\cite{ChanUMCubes} , \cite{Godsil-Mullin-Roy}).

We adopt the notational convention of denoting  a sum in the group algebra of a subset $Y\subseteq G$ by the same symbol $Y$. For a character $\chi$ of $G$
we will also write $\chi(Y)$ for $\sum_{y\in Y}\chi(y)$.

Let  $S$ be the skew adjacency matrix. We can view $S$
as the regular representation of the group algebra element $C-C^{(-1)}$.
Now $iS$ is hermitian so
$e^{-it(iS)}=e^{tS}$ is a real orthogonal matrix. It follows that $UM$
occurs at time $\tau$ if and only if there exist $t_j\in\{1,-1\}$, $j=1$,\dots,
$k$, such that
\begin{equation}\label{had_eq}
\sqrt{G}e^{\tau S}=\sum_{j=1}^k t_jK_j.
\end{equation}
This is because if \eqref{had_eq} holds then the matrix on the right with
$\pm 1$ entries must be a real Hadamard matrix while,  conversely, every matrix of the form $\sqrt{G}e^{tS}$ lies in the span of the $K_j$, and its entries are all $\pm 1$ only if it has the form of the right side of \eqref{had_eq}.

The eigenvalues of $\Cay(G,C)$ are given by the irreducible complex
characters $\{\chi_r\}_{r=1}^k$ of $G$. The character $\chi_r$
contributes the eigenvalue
\begin{equation}
  \theta_r=\frac{\chi_r(C)-\chi_r(C^{(-1)})}{\chi_r(1)}=\frac{\chi_r(C)-\overline{\chi_r(C})}{\chi_r(1)}
\end{equation}
with multiplicity $\chi_r(1)$. Different characters may give the same
eigenvalue, whose total multiplicity will then be the sum of the corresponding character degrees.

Using the  spectral decomposition, it follows that \eqref{had_eq} holds if and only if for all $r=1$,\dots,$k$, we have
\begin{equation}\label{char_eq}
\sqrt{G}e^{\tau \theta_r}=\sum_{j=1}^k t_j\frac{\chi_r(K_j)}{\chi_r(1)},
\end{equation}
and it is \eqref{char_eq} that we shall use from now on.

\section{Cayley graphs on Suzuki $2$-groups}\label{groups_graphs}
We now describe a class of nonabelian (normal)
oriented Cayley graphs that have UM.

Let $n=2m+1$ be an odd positive integer. Let $\theta$ 
be any generator of the cyclic group $\gal(\F_{2^n}/\F_2)$.
Let $A(n,\theta)$ denote the group of matrices 
\begin{equation}\label{matrix}
  \begin{bmatrix}1&a&b\\0&1&a^\theta\\0&0&1
  \end{bmatrix}, \quad \text{$a\in \F_{2^n}$.}
\end{equation}
These are called Suzuki $2$-groups of type $A$, first studied by G. Higman
\cite{Higman}
Note that for $n=1$, $A(1,\theta)\cong\Z_4$.
For $n>1$, let $G=A(n,\theta)$. For simplicty of notation, we  denote the above element \eqref{matrix} of $G$ by the ordered $(a,b)$. Then the following
formulas can be checked directly.
  \begin{enumerate} \label{gp_prods}
  \item[(i)] $(a,b)(c,d)=(a+c,b+d+ac^\theta)$;
  \item[(ii)] $(a,b)^{-1}=(a,b+aa^\theta)$;
  \item[(iii)] $(a,b)^{-1}(c,d)(a,b)=(c,d+ac^\theta+a^\theta c)$;
\item[(iv)]$[(c,d),(a,b)]:= (c,d)^{-1}(a,b)^{-1}(c,d)(a,b)=(0,ac^\theta+a^\theta c)$.
  \end{enumerate}

  It is easy to verify that $Z(G)=\{(0,b)\mid b\in \F_{2^n}$ and $Z(G)$ contains all of the involutions in $G$, and elements outside $Z(G)$ have order $4$,
  so $G/Z(G)$ and $Z(G)$ can both be viewed as $n$-dimensional vector spaces
  over $\F_2$.
  
  An important property of $G$ is the existence of a cylic group of
  automorphisms that acts regularly on the set of involutions.
  Using formula (i), we see that if $\lambda$ is a generator of $\F_{2^n}^\times$,
  then the map $\xi:G\to G$ given by
  \begin{equation}\label{xi}
    \xi(a,b)=(\lambda a,\lambda^{1+\theta}b.
  \end{equation}
  is an automorphism of order $2^n-1$. The group $\langle xi\rangle$ acts
  regularly on the set of nonzero elements of $Z(G)$, and also
  on the nonzero elements of $G/Z(G)$ in its induced action.

We summarize some further facts about
$G$ in the following lemma. A group $E$ of prime power order is {\it extraspecial} if its center has prime order and its quotient by the center is elementary
abelian (and nontrivial).
\begin{lemma}\label{gp_properties}
  \begin{enumerate}
     \item[(a)] $Z(G)$ is equal to the commutator subgroup $[G,G]$ of $G$.
  \item[(b)] If $x\in G$ has order $4$, its centralizer is
    $C_G(x)=\langle x,\Z(G)\rangle$.
  \item[(c)] For each noncentral element $x$ the subgroup
    $[x,G]=\langle x^{-1}g^{-1}xg\mid g\in G\rangle$ is a maximal
    subgroup of $Z(G)$, and every maximal subgroup of $Z(G)$ has this form
    with $[x,G]=[x',G]$ iff $xZ(G)=x'Z(G)$.
  \item[(d)] The conjugacy class of a noncentral element $x$ is equal
    to the coset $x[x,G]$ and has size $2^{n-1}$.
  \item[(e)] No element of order $4$ is conjugate to its inverse. There are $2(2^n-1)$ conjugacy classes of elements of order $4$.
\item[(f)] Let $M_x:=[x,G]$, for $x\in G$ of order $4$.  Then $G/M_x$ is isomorphic to a central product $\langle x\rangle\Ydown E(x)$ of $\langle x\rangle\cong C_4$ with an extraspecial group $E(x)$ of order $2^n=2^{2m+1}$, amalgamating the subgroups  $\langle x^2\rangle$ and  $Z(E(x))$ or order $2$.
  \end{enumerate}
\end{lemma}
\begin{proof} Since $G/Z(G)$ is abelian we have $[G,G]\le Z(G)$. As $G$ is
  $[G,G]$ is a characteristic subgroup and contains a nonzero element of
  $Z(G)$ by formula (iv). Thefore, by transitivity of $\langle\xi\rangle\>$ on the nonzero elements of $Z(G)$ we have $Z(G)\leq [G,G]$.
  Part (b) follows from formula (iii).
  For  $x\in G\setminus Z(G)$ define $f_x(gZ(G))= [x,g]$.
  By (a), $f_x:G/Z(G)\to Z(G)$ is a well-defined linear map of
  $\F_2$-vector spaces. The image of $f_x$ is equal to $[x,G]$. By (c), the kernel of $f_x$ is  $\langle xZ(G)\rangle$ which is a one-dimensional.
  Hence $[x,G]$ is a hyperplane of $\F_2$-vector space $Z(G)$.
  To see that every maximal subgroup of $Z(G)$ has the form $[x,G]$,
  we note first that for any automorphism $\eta$ of $G$, we have
  $\eta([x,G])=[\eta(x),G]$. Since the group $\langle\xi\rangle$
  acts regularly on the nonzero elements of $Z(G)$, it also acts regularly
  on the set of hyperplanes. This every hyperplane has the form
  $\xi^e([x,G])=[\xi^e(x),G]$, for some $e$. It is clear that $[x,G]=[x',G]$
  when $xZ(G)=x'Z(G)$. Therefore we have a map from
  the from $(G/Z(G))\setminus\{0\}$ to the set of hyperplanes of $Z(G)$
  given by $xZ(G)\mapsto [x,G]$. We have proved that this map is surjective.
  As the domain and codomain have the same size, the map is bijective,
  and (c) is prove.  Parts (d) is immediate from (c).

  To prove (e) we observe, as in \cite[p.82]{Higman},
  that since $\theta$ has odd order $r$, the map $a\mapsto a^{1+\theta}$
  is invertible (because $2=1+\theta^r$ can be factored $(1+\theta)\psi$,
  and the squaring map $2$ is invertible). We use this to
  show that if $a\neq0$, then  $x=(a,b)$ is not conjugate to its inverse.
  For if this were so, we see from formulas (ii) and (iii)  that
  there would be some $c\in\F_{2^n}$ such that
  $ac^\theta+a^\theta c=aa^\theta$, 
  which can be rewritten as $(a+c)^{1+\theta}=c^{1+\theta}$. This contradicts
  the bijectivity of $(1+\theta)$, and so (e) is proved.
  A similar calculation shows that $x^2\notin [x,G]$, so the image
  of $x$ in $G/M_x$ has order $4$, which we will now use in the proof of (f).
  Let  $\{\overline y_s\}_{s=1}^n$ be a basis of $G/Z(G)$, with $y_1$ the image of $x$. Let $y_i$ be preimages in $G$ and let $\tilde E(x)\leq G/M_x$ be the
  image of subgroup of $G$ generated by $y_2$,\dots,$y_n$ and $Z(G)$,
  and $E(x)$ its image in $G/M_x$. Since $G/Z(G)$ is elementary abelian, 
  so is $E(x)/Z(E(x))$.
  We claim that $E(x)$ is extraspecial. It remains to show
  that $Z(E(x))$ is cyclic of order 2. Let $y\in\tilde E(x)$, with $y\notin
  Z(G)$.
  Then $yZ(G)\neq xZ(G)$, so $M_y\neq M_x$. Thus, there exists
  $w\in G$ such that $f_y(wZ(G))\notin[x,G]$.  Since $f_y(wxZ(G))=f_y(wZ(G))f_y(xZ(G))$ and $f_y(xZ(G))\in M_x$, we may take $w\in \tilde E(x)$.
  This means that the image of  $y$ in $E(x)$ does not lie in the $Z(E(x))$.
  Thus $Z(E(x))=Z(G)/M_x$, which is cyclic of order 2, and we have proved that
  $E(x)$ is extraspecial. It is now easy to check that $G/M_x$ is the central
  product of $E(x)$ with $\langle x\rangle$, with the $Z(E(x))=\langle x^2\rangle$.
\end{proof}

Let $C$ be the union of exactly half of the conjugacy classes of elements
of order $4$, chosen so that no element of $C$ has its inverse in $C$.
By parts (b) and (d) of Lemma~\ref{gp_properties} $\abs{C}=2^{n-1}(2^n-1)$. There are, of course, many ways to choose $C$.
\begin{theorem} With $G=A(n,\theta)$ and $C$ as above, $\Cay(G,C)$
  has uniform mixing at time $\tau=\pi/2^{n+1}$. 
\end{theorem}\label{mainthm}
The proof of the theorem consists of verifying the condition given in
\eqref{char_eq} for an appropriate choice of signs $t_j$. To this end,
we need  a complete description of the irreducible characters of $G$,
which is the subject of the next section.

\section{The irreducible characters of $A(n,\theta)$}\label{chars}
Since a group with a faithful irreducible complex representation must
have a cyclic center, it follows that every irreducible complex representation
of $G$ must contain a maximal subgroup of $Z(G)$ in its  kernel.
The maximal subgroups of $Z(G)$ have the form $M_x=[x,G]$ for some $x\in
G\setminus Z(G)$.  Then $G(x):=G/M_x$ the central product
$\langle x\rangle\Ydown E(x)$ of
$\langle x\rangle\cong C_4$ with an extraspecial group $E(x)$ of order
$2^n=2^{2m+1}$, amalgamating the subgroups  $\langle x^2\rangle$ and  $Z(E(x))$ or order $2$. We refer to \cite[Chapter 5.5]{G3} for the  necessary background information on extraspecial groups and their irreducible representations.
From the character theory of extraspecial groups, it follows that
$G(x)$ has exactly $2$ nonlinear irreducible
complex characters $\chi_x^+$ and $\chi_x^-$, whose
restrictions to $E(x)$ are both equal to the unique nonlinear irreducible character of degree $2^m$ of $E(x)$, and such that $\chi_x^+(x)=2^mi$ and $\chi_x^-(x)=-2^mi$. In the usual way, the characters $\chi_x^\pm$ can be viewed as
characters of $G$ having $M_x$ in their kernel.
If we allow $x$ to vary over a set of representative generators of
the $2^n-1$ cyclic subgroups of order $4$ in $G$, we obtain
$2(2^n-1)$ irreducible characters of degree $2^m$ and they are distinct
since the characters $\chi_x^\pm$ take the value $0$ outside $Z(G(x))=\langle x\rangle$.
The characters $\chi_x^{\pm}$, together with the $2^n$ linear characters of $G$, we have all the irreducible characters since the sum  of their squared degrees equals $2^n+2(2^n-1)2^{2m}=2^{2n}=\abs{G}$.

\section{Proof of Theorem~\ref{mainthm}}\label{proofs}
We can now complete the proof of Theorem~\ref{mainthm}.
We set
\begin{equation}\label{coeffs}
  t_j=\begin{cases}1\qquad\text{if $K_j\subseteq C$;}\\
  -1\qquad\text{if $K_j\subseteq C^{(-1)}$;}\\
  1\qquad\text{if $K_j$ is a central.}
  \end{cases}
\end{equation}
With this choice of signs $t_j$ it is now a simple matter  to verify, using the description of the irreducible character values given above,
that the equation \eqref{char_eq} holds for every irreducible character when
$\tau=\frac{\pi}{2^{n+1}}$.

Consider first  a linear character $\lambda$. Then, as  $Z(G)$ equals the
commutator subgroup, it lies in the kernel of $\lambda$, which can therefore
be considered a character of $G/Z(G)$. Consider a class $K_j$ of elements of order 4 and let $K_{j^*}$ denote the inverse
class. All elements in both classes map to a single element of $G/Z(G)$, so $\lambda$ is constant on $K_j\cup K_{j^*}$.
Since $t_j=-t_{j^*}$, we see that the contribution from $K_j\cup K_{j^*}$ to the right side of \eqref{char_eq} is zero.
It follows that for $\lambda$ the  right side of \eqref{char_eq} is $\abs{Z(G)}=2^n$
The same reasoning also shows that the eigenvalue corresponding to $\lambda$ is equal to $0$.
Therefore \eqref{char_eq} is satisfied for all linear characters $\lambda$.

Now consider one of the nonlinear characters. As the cases are similar
we just  consider $\chi_x^-$.
Suppose $z\in Z(G)$. Then $\chi_x^-(z)=2^m$  if $z\in M_x$
and $\chi_x^-(z)=-2^m$ if $z\notin M_x$. Therefore, the central elements contribute zero to the right had side of \eqref{char_eq}.
As for elements $y$ outside $Z(G)$,
which all have order $4$, we have $\chi_x^-(y)=0$ unless $y$ is conjugate
to $x$ or $x^{-1}$, with  $\chi_x^-(x)=-2^mi$ and $\chi_x^-(x^{-1})=2^mi$. 
Since the conjugacy classes of $x$ and $x^{-1}$  have size $2^{n-1}$, we see that the right
side of \eqref{char_eq} for $\chi_x^-$ evaluates to $\frac{2^{n-1}(-2^m-2^m)}{2^m}=-2^ni$. The corresponding eigenvalue is
\begin{equation}
\theta_x^-:=\frac{\chi_x^-(C)-\overline\chi_x^-(C)}{2^m}=-2^ni.
\end{equation}
Similarly the right side of \eqref{char_eq} for $\chi_x^+$ evaluates to $2^ni$
and we have $\theta_x^+=2^ni$.
In both cases \eqref{char_eq} holds when $\tau=\frac{\pi}{2^{n+1}}$.

\qed

\begin{concremark*} The sets $C$ we took as connection sets for our
  oriented normal Cayley graphs are somewhat remarkable as they have previously 
  appeared in a different guise in the paper \cite{Gow-Quinlan} of Gow and Quinlan, as examples of  {\it central difference sets}. Here, ``central'' means that the difference set is a
  union of conjugacy classes, so corresponds to `` normal'' for Cayley graphs.   It appears that central difference sets in nonabelian groups are extremely   rare; no other examples are  known to the author at this time.
\end{concremark*}

\section{Acknowledgements}
I would like to thank  Ada Chan, Chris Godsil and Raghu Pantangi for several helpful conversations related to the topic of this note.

\end{document}